\documentclass[pra,reprint,twocolumn,nofootinbib,superscriptaddress,nopacs,aps,10pt]{revtex4-1}
\usepackage{amsmath}
\usepackage{hyperref}
\usepackage{graphicx}
\usepackage{amsfonts}
\usepackage{mathrsfs} 
\usepackage{amsthm}
\usepackage{amsmath}
\usepackage{amssymb}
\usepackage{bbm} 
\usepackage{mathtools} 
\newcommand{\ket}[1]{\vert #1 \rangle}
\newcommand{\bra}[1]{\langle #1 \vert}

\newcommand\proj[1]{\vert #1 \rangle \langle #1 \vert}

\newcommand{\opn}[1]{\operatorname{#1}}
\DeclareMathOperator{\tr}{Tr}  
\newcommand{\renyi}{R\'enyi~}
\newcommand{\density}[1]{\mathscr{D}(#1)}
\newcommand{\pos}[1]{\mathscr{P}(#1)}
\newcommand{\linear}[1]{\mathscr{L}(#1)}

\newcommand{\1}{\mathbbm{1}}

\DeclareMathOperator{\Collision}{H_{2}}
\newcommand{\Ppg}{\operatorname{P}^{\operatorname{pg}}}
\newcommand{\Pro}{\operatorname{P}}
\newcommand{\Fid}{\operatorname{F}}
\newcommand{\Fpg}{\operatorname{F}^{\operatorname{pg}}}

\newcommand*{\cB}{\mathcal{B}}

\newcommand*{\cH}{\mathcal{H}}

\newcommand*{\cM}{\mathcal{M}}

\newcommand*{\cP}{\mathcal{P}}
\newcommand*{\cQ}{\mathcal{Q}}
\newcommand*{\cR}{\mathcal{R}}

\newcommand*{\bF}{\mathbb{F}}

\newtheorem{theorem}{Theorem}
\newtheorem{lemma}[theorem]{Lemma}

\begin{document}

\title{Uncertainty Relations in the Presence of Quantum Memory for Mutually Unbiased Measurements}
\author{Kun Wang}
\email{wk@smail.nju.edu.cn}
\author{Nan Wu}
\email[Corresponding author:~]{nwu@nju.edu.cn}
\author{Fangmin Song}
\email[Corresponding author:~]{fmsong@nju.edu.cn}
\affiliation{State Key Laboratory for Novel Software Technology,
            Department of Computer Science and Technology, Nanjing University, Nanjing 210093, China}
\date{\today}
\begin{abstract}
In~\cite{berta2014entanglement}, uncertainty relations in the presence of quantum memory was formulated for mutually
unbiased bases using conditional collision entropy. In this paper, we generalize their results to the mutually unbiased
measurements. Our primary result is an equality between the amount of uncertainty for a set of measurements and the
amount of entanglement of the measured state, both of which are quantified by the conditional collision entropy.
Implications of this equality relation are discussed. We further show that similar equality relation can be obtained
for generalized symmetric informationally complete measurements. We also derive an interesting equality for arbitrary
orthogonal basis of the space of Hermitian, traceless operators.
\end{abstract}

\maketitle

\section{Introduction}


Uncertainty relations form a central part of quantum mechanics. They impose fundamental limitations on our ability to
simultaneously predict the outcomes of noncommuting observables. Different approaches have been proposed to quantify
these relations. The original formulation is given by Heisenberg~\cite{heisenberg1927w} in terms of standard deviations
for momentum and position operators. His result is then generalized to two arbitrary
observables~\cite{robertson1929uncertainty}.
Later it is recognized that one can express uncertainty relations in terms of 
entropies~\cite{bialynicki1975uncertainty,deutsch1983uncertainty,maassen1988generalized}. 
In this approach, entropy functions like the Shannon and \renyi entropies are
used to quantify uncertainty (Ref.~\cite{coles2017entropic} is a nice survey on this topic).

Mutually unbiased bases (MUB) have many applications in quantum information theory: quantum error correction
codes~\cite{spengler2013graph}, quantum cryptography~\cite{cerf2002security}, and entanglement
detection~\cite{spengler2012entanglement} (see review~\cite{durt2010mutually} and references therein). There has been
of great effort and research interest in constructing the complete set of MUB. However, The existence problem of
complete set of mutually unbiased bases for arbitrary dimension is still open. In~\cite{kalev2014mutually}, the authors
proposed the concept of mutually unbiased measurements (MUM). These measurements contain the complete set of MUBs as a
special case while the measurement operators need not be rank one projectors. They proved that a complete set of
mutually unbiased measurements can be built explicitly for arbitrary finite dimension.


Uncertainty relations in the presence of quantum memory was formulated for MUBs using conditional collision entropy
in~\cite{berta2014entanglement}, in which the authors gave an exact relation between the amount of uncertainty as
measured by the guessing probability and the amount of entanglement as measured by the recoverable entanglement
fidelity. As MUMs are natural generalizations of MUBs, one may naturally conjecture that similar uncertainty relations
hold for MUMs. In this paper, we show that this is indeed the case: we generalize their results to the set of MUMs. The
main result is an equality between the amount of uncertainty for a set of measurements and the amount of entanglement
of the measured state, both of which are quantified by the conditional collision entropy. 

The rest of this paper is organized as follows. In Sec.~\ref{sec:preliminaries}, we establish the notation and briefly
review the concepts of MUMs, conical 2-design, and conditional collision entropy. In
Sec.~\ref{sec:mum}, we present our central result --- 
an equality between the amount of uncertainty
for a complete set of MUMs and the amount of entanglement of the measured state, both of which are
quantified by the conditional collision entropy.
We discuss several implications of this
equality relation. We further show that an equality relation can be obtained for generalized symmetric
informationally complete measurements in Sec.~\ref{sec:sim}. We conclude in Sec.~\ref{sec:conclusion}. Some proofs are
given in the Appendix.
\section{Preliminaries}\label{sec:preliminaries}

A quantum system $A$ is associated to a Hilbert space
$\cH_A$ with some fixed orthonormal basis $\{\ket{s}\}$.
Throughout this article, we assume $\cH_A$ is $d$-dimensional.
If the underlying system is clear from context, we simply write the space as $\cH$.
We denote by $\linear{\cH}$ the set of linear operators,
by $\pos{\cH}$ the set of positive semidefinite operators, and by $\density{\cH}$ the set of density operators on
$\cH$. We use $\1_A$ to represent the identity operator and $\pi_A$ to represent the maximally mixed operator of system
$A$. Systems with the same letter are assumed to be isomorphic: $A^\prime \cong A$. We denote by
$\ket{\Psi_{AA^\prime}}$ the normalized maximally entangled state on system $AA^\prime$, which has the form
$\ket{\Psi_{AA^\prime}} = \sum_s \ket{ss}$. For simplicity, we let $[d] = \{1, \cdots, d\}$.

\subsection{Mutually unbiased measurements}
Two orthonormal bases $\cB^{(1)} = \{\ket{\psi_{x}^{(1)}}\}_{x\in[d]}$ and
$\cB^{(2)} = \{\ket{\psi_{x}^{(2)}}\}_{x\in[d]}$ of $\cH$ are said to be \textit{mutually unbiased} if
\[
\vert \langle \psi_{x}^{(1)} \vert \psi_{y}^{(2)} \rangle \vert = \frac{1}{\sqrt{d}},
\forall x,y \in [d].
\]
Intuitively, if $\cH$ is prepared in an eigenstate of $\cB^{(1)}$ and measured in $\cB^{(2)}$, the measurement outcome
is completely random. A set of orthonormal bases $\{\cB^{(\theta)}\}_{\theta\in\Theta}$ forms a set of MUBs if these
bases are pairwise unbiased. In a $d$-dimensional Hilbert space there are at most $d+1$ pairwise unbiased
bases~\cite{durt2010mutually}. This set is called a \textit{complete set} of MUBs. It is open whether complete set of
MUBs exists for arbitrary $d$.

By generalizing the notion of ``unbiasedness'', the concept of mutually unbiased measurements (MUM) is
introduced~\cite{kalev2014mutually}. Two POVMs $\cP^{(1)} = \{P_{x}^{(1)}\}_{x\in[d]}$ and $\cP^{(2)} =
\{P_{x}^{(2)}\}_{x\in[d]}$ are mutually unbiased if the following conditions are satisfied
for all $x,x^\prime\in[d], \theta=1,2$:
\begin{align*}
\tr\left[ P_{x}^{(\theta)} \right] &=  1,  \\ 
\tr\left[ P_{x}^{(1)}P_{x}^{(2)} \right] &=  \frac{1}{d}, \\ 
\tr\left[ P_{x}^{(\theta)}P_{x^\prime}^{(\theta)} \right] &=  
\delta_{x,x^\prime}\kappa + (1-\delta_{x,x^\prime})\frac{1-\kappa}{d-1}, 
\end{align*}
where the \textit{efficiency parameter} $\kappa$ satisfies $1/d < \kappa \leq 1$. $\kappa$ determines how close the
measurements operators are to rank-one projectors: $\kappa = 1$ if and only if $\cP^{(1)}$ and $\cP^{(2)}$ form two
MUBs. Unlike the existence problem of a complete set of MUBs, there exists a general construction of complete set of
MUMs for arbitrary finite $d$~\cite{kalev2014mutually}. Let $\{F_k\}_{k\in[d^2-1]}$ be an orthogonal basis for the
space of Hermitian, traceless operators acting on $\cH$. We regard these operators as elements of a $(d+1)\times(d-1)$
block matrix
\[
\begin{pmatrix}
  F_1 & F_2 & \cdots & F_{d-1} \\
  F_{d} & F_{d+1} & \cdots & F_{2(d-1)} \\
  \vdots & \vdots & \ddots & \vdots \\
  F_{d(d-1)+1} & F_{d(d-1)+2} & \cdots & F_{(d+1)(d-1)}
\end{pmatrix}.
\]
We relabel the block matrix by a tuple $(x,\theta):x\in[d-1],\theta\in[d+1]$ based on their (column, row) location
\[
\begin{pmatrix}
  F_{1,1} & F_{2,1} & \cdots & F_{d-1,1} \\
  F_{1,2} & F_{2,2} & \cdots & F_{d-1,2} \\
  \vdots & \vdots & \ddots & \vdots \\
  F_{1,d+1} & F_{2,d+1} & \cdots & F_{d-1,d+1}
\end{pmatrix}.
\]
Based on $\{F_{x,\theta}\}$, define the following $d(d+1)$ operators 
\begin{equation}\label{eq:f_x_theta}
F_x^{(\theta)} =
  \begin{cases}
    F^{(\theta)} - d(d+\sqrt{d})F_{x,\theta}, &x\in[d-1], \\
    (1+\sqrt{d})F^{(\theta)}, & x = d,
  \end{cases}
\end{equation}
where $F^{(\theta)} = \sum_{x=1}^{d-1} F_{x,\theta}$. Then, the operators
\begin{equation}\label{eq:mum-operator}
  P_{x}^{(\theta)} = \frac{1}{d}\1 + tF_x^{(\theta)}, x\in[d], \theta\in[d+1]
\end{equation}
form a complete set of MUMs, with the parameter $t$ chosen such that $P_{x}^{(\theta)} \geq 0$.
The efficiency parameter $\kappa$ is then given by
\begin{equation}\label{eq:t-kappa}
  \kappa = \frac{1}{d} + t^2\left(1+\sqrt{d}\right)^2(d-1).
\end{equation}

\subsection{Conical 2-design}

A \textit{complex projective 2-design} is a set of vectors $\{\ket{\psi_x}\}_{x\in \Sigma}$ 
(not necessarily normalized) lying in $\cH_A$ such that~\cite{berta2014entanglement}
\[
\frac{1}{\vert \Sigma \vert}\sum_{x \in \Sigma} \proj{\psi_x}\otimes\proj{\psi_x}
= \frac{1}{d(d+1)} \left( \1_{AA^\prime} + \bF_{AA^\prime} \right),
\]
where $\bF_{AA^\prime}$ is the swap operator defined as
\begin{equation}\label{eq:swap-operator}
\bF_{AA^\prime} = \sum_{s,t} \ket{s}\bra{t} \otimes \ket{t}\bra{s}.
\end{equation}
Complex projective designs play an important role in quantum information theory.
A best known example is the complete set of MUBs.
Let $\{\cB^{(\theta)}\}_{\theta\in[d+1]}$ be a complete set of MUBs (if exists) on $\cH_A$. 
It is proved in~\cite{klappenecker2005mutually} that such a set generates a complex projective 2-design:
\begin{equation}\label{eq:mub-two-design}
\sum_{\theta=1}^{d+1}\sum_{x=1}^{d} \proj{\psi_{x}^{(\theta)}} \otimes \proj{\psi_{x}^{(\theta)}}
= \1_{AA^\prime} + \bF_{AA^\prime}.
\end{equation}

A complex projective $2$-design consists of rank-one projectors. In~\cite{graydon2016quantum} the authors introduce a
generalization which shares properties with complex projective $2$-design, but in which the projectors are arbitrary
positive semi-definite operators. A \textit{conical 2-design} is a set of positive semidefinite operators
$\{A_x\}_{x\in\Sigma}$ in $\cH_A$ satisfying
\[
\sum_{x\in\Sigma} A_x \otimes A_x = k_+ \1_{AA^\prime} + k_{-} \bF_{AA^\prime}
\]
for some $k_{+} \geq k_{-} \geq 0$.
As MUMs are generalizations of MUBs, we wish similar property (that complete set of MUBs forms a complex projective
2-design) holds for MUMs. In~\cite{graydon2016quantum}, it is proved that a complete set of MUMs forms a conical
2-design
\begin{equation}\label{eq:mum-generalization}
\sum_{\theta=1}^{d+1}\sum_{x=1}^d P_{x}^{(\theta)} \otimes P_{x}^{(\theta)}
= f(\kappa)\1_{AA^\prime} + g(\kappa)\bF_{AA^\prime},
\end{equation}
where the coefficients are given by
\begin{equation}\label{eq:coefficients}
f(\kappa) =  1+ \frac{1-\kappa}{d-1}, \quad
g(\kappa) = \frac{\kappa d - 1}{d-1}.
\end{equation}
Eq.~\eqref{eq:mum-generalization} can be viewed as a generalization of Eq.~\eqref{eq:mub-two-design}.

\subsection{Conditional collision entropy}

We use conditional collision entropy as measure of uncertainty. Let $\rho_{AB}\in\density{\cH_A\otimes\cH_B}$ be a
quantum state, the \textit{conditional collision entropy} is defined as~\cite{tomamichel2015quantum}
\begin{equation}
\Collision(A{\vert}B)_\rho 
= -\log\tr[\rho_{AB}(\1_A\otimes\rho_B)^{-1/2}\rho_{AB}(\1_A\otimes\rho_B)^{-1/2}].
\end{equation}
Trivializing system $B$, we get the collision entropy of single system:
$\Collision\left(A\right)_\rho = -\log\tr\rho_A^2$.

Collision entropy admits nice operational interpretations. Let $\rho_{XB} = \sum_x \eta_x \proj{x} \otimes \rho_x$ be a
classical quantum state shared between Alice and Bob. From Bob's view, he owns a state ensemble $\{\eta_x, \rho_x\}$.
He may use the \textit{pretty good} measurement $\cM^{\opn{pg}} = \{M_y\}$ associated with $\rho_{XB}$ to extract
information about index $x$. The measurement operators are given by
$M_y = \rho_B^{-1/2} (\eta_y\rho_y) \rho_B^{-1/2}$,
where $\rho_B = \tr_X \rho_{XB} = \sum_x \eta_x \rho_x$.
Denote by $\Pro^{\opn{pg}}\left(X \vert B \right)_\rho$ the probability
that he can correctly guess the index $x$ on average, then
\[
\Pro^{\opn{pg}}\left(X \vert B \right)_\rho = \sum_x \eta_x \tr\left[M_x \rho_x\right].
\]
It is proved in~\cite{buhrman2008possibility} 
that $\Collision(X \vert B )_\rho$ has the following operational interpretation:
\begin{equation*}
\Pro^{\opn{pg}}\left(X \vert B \right)_\rho = 2^{- \Collision(X \vert B )_\rho}.
\end{equation*}
Now we consider the fully quantum conditional collision entropy. Given state $\rho_{AB}$, the \textit{pretty-good
recoverable entanglement fidelity} quantifies how well 
the local pretty-good recovery map $\cR^{\opn{pg}}_{B\to A^\prime}$ (defined in~\cite{berta2014entanglement})
can bring $\rho_{AB}$ to $\ket{\Psi_{AA^\prime}}$:
\begin{equation*}
  \Fid^{\opn{pg}}\left(A \vert B\right)_\rho
  = d_A \Fid\left( (\1_A\otimes\cR^{\opn{pg}}_{B\to A^\prime})\rho_{AB}, \proj{\Psi_{AA^\prime}}\right),
\end{equation*}
where $\Fid(\rho,\sigma)=(\tr\sqrt{\sqrt{\sigma}\rho\sqrt{\sigma}} )^2$ is Uhlmann's
fidelity~\cite{nielsen2011quantum}. It is proved in~\cite{berta2014entanglement} that $\Collision(A \vert B )_\rho$ has
the following operational interpretation:
\begin{equation*}
\Fid^{\opn{pg}}\left(A \vert B\right)_\rho = 2^{- \Collision(A \vert B )_\rho}.
\end{equation*}

\section{Uncertainty relations for complete set of MUMs}\label{sec:mum}

In this section, we present uncertainty relations in the presence of memory for a complete set of mutually unbiased
measurements. The main result is an equality quantifying the relation between uncertainty and entanglement,
both of which are measured by conditional collision entropy.

Let $\cP^{(\theta)} = \{P_{x}^{(\theta)}\}_{x\in[d]}$ be a MUM in $A$
and $\rho_{AB}$ be quantum state on $AB$.
Measuring $\rho_{AB}$ on $A$ by $\cP^{(\theta)}$, we arrive at a classical-quantum state
\begin{equation}\label{eq:omega-x-theta-b}
\omega_{X^{(\theta)}B} = \sum_{x=1}^{d} 
\proj{x} \otimes \tr_A\left[ \left(P_x^{(\theta)} \otimes \1\right) \rho_{AB} \right],
\end{equation}
where $\cH_X$ is a $d$-dimensional Hilbert space with $\{\ket{x}\}$ being its standard basis.
The classical register $X$ indicates which measurement operator is performed;
$\tr_A[P_x^{(\theta)}\rho_{AB}]$ is the post-measurement state 
(unnormalized) left in system $B$, conditioned on the measurement operator performed;
and $\tr[P_x^{(\theta)}\rho_{AB}]$ is probability that the measurement outcome is $x$.
We remark that the choice of $\{\ket{x}\}$ does not affect our result as long as it forms an 
orthonormal basis of $\cH_X$.

Uncertainty relations study the unpredictability about the outcomes of many incompatible measurements.
Thus in the following, we will not only measure in one fixed MUM, but with equal probability in one of $d + 1$ MUMs. 
Let $\{\cP^{(\theta)}\}_{\theta\in[d+1]}$ be a complete set of MUMs on system $A$, we define the following 
classical-quantum state
\begin{align}
\omega_{XB\Theta} &= \frac{1}{d+1}\sum_{\theta=1}^{d+1} \sum_{x=1}^{d} \proj{x}_X \notag \\
&~ \quad\otimes \tr_A\left[ \left(P_x^{(\theta)} \otimes \1_B\right) \rho_{AB} \right] \otimes \proj{\theta}_\Theta,
\label{eq:cq-state}
\end{align}
where $\Theta$ is an indicator specifying which measurement has been performed.
The collision entropy of $\omega_{XB\Theta}$, with partition $X{:}B\Theta$, can be expressed as
\begin{equation}
\Collision\left(X\vert B\Theta\right)_\omega
=- \log\left( 
\frac{1}{d+1}\sum_{\theta,x}
\tr_B\left\{\tr_A[P_x^{(\theta)}\widetilde{\rho}_{AB}]^2\right\}\right),
\label{eq:classical-quantum-collision}
\end{equation}
where $\widetilde{\rho}_{AB} = \rho_B^{-1/4}\rho_{AB}\rho_B^{-1/4}$.
For the proof of Eq.~\eqref{eq:classical-quantum-collision}, see Appx.~\ref{appx:classical-quantum-collision}.
Under this convention, the conditioned collision entropy of $\rho_{AB}$ can be rewritten as
$\Collision\left(A {\vert} B\right)_\rho = -\log\tr[\widetilde{\rho}_{AB}^2]$.
We are now ready to present our main result.
\begin{theorem}\label{thm:mum-ur}
Let $\{\cP^{(\theta)}\}_{\theta\in[d+1]}$ be a complete set of MUMs on system $A$.
For arbitrary quantum state $\rho_{AB}$, it holds that
\begin{equation}\label{eq:complete-eq}
\Collision\left(A\vert B\Theta\right)_\omega
= \log\left(d + 1 \right) 
- \log\left(f(\kappa) + g(\kappa)2^{-\Collision\left(A\vert B\right)_\rho} \right),
\end{equation}
where where $\omega_{XB\Theta}$ is defined in Eq.~\eqref{eq:cq-state},
$f(\kappa)$ and $g(\kappa)$ are defined in Eq.~\eqref{eq:coefficients}.
\end{theorem}
\begin{proof}
The proof is similar to the proof outlined in \textbf{Appendix B} of~\cite{berta2014entanglement}.
We introduce the spaces $A^\prime \cong A$ and $B^\prime \cong B$,
as well as the state $\widetilde{\rho}_{A^\prime B^\prime} \cong \widetilde{\rho}_{AB}$. Then we have
{
\small
\begin{align}
&\quad\; (d+1)2^{-\Collision\left(X\vert B\Theta\right)_\omega}\notag\\
&= \sum_{\theta,x}
    \tr_B\left\{  \tr_A[P_x^{(\theta)}\widetilde{\rho}_{AB}]^2 \right\}\notag\\
&= \sum_{\theta,x}\tr_{BB^\prime}\tr_{AA^\prime}
    \left[ \left(P_x^{(\theta)} \otimes P_x^{(\theta)}\right)
            \left(\widetilde{\rho}_{AB} \otimes \widetilde{\rho}_{A^\prime B^\prime} \right) \bF_{BB^\prime} \right] 
    \notag\\
&= \tr_{BB^\prime}\tr_{AA^\prime}
    \left[ \left( \sum_{\theta,x} P_x^{(\theta)} \otimes P_x^{(\theta)}\right)
            \left(\widetilde{\rho}_{AB} \otimes \widetilde{\rho}_{A^\prime B^\prime} \right) \bF_{BB^\prime} \right] 
    \notag\\
&= \tr_{BB^\prime}\tr_{AA^\prime}
    \left[ \left( f(\kappa)\1_{AA^\prime} + g(\kappa) \bF_{AA^\prime}\right)
    \left(\widetilde{\rho}_{AB} \otimes \widetilde{\rho}_{A^\prime B^\prime} \right) \bF_{BB^\prime} \right] 
    \notag\\
&= f(\kappa) \tr_{BB^\prime}\tr_{AA^\prime}
    \left[ \left(\widetilde{\rho}_{AB} \otimes \widetilde{\rho}_{A^\prime B^\prime} \right) \bF_{BB^\prime} \right] 
    \label{eq:term1}\\
&+ g(\kappa) \tr_{BB^\prime}\tr_{AA^\prime}
    \left[ \bF_{AA^\prime} \left(\widetilde{\rho}_{AB} \otimes \widetilde{\rho}_{A^\prime B^\prime} \right) 
            \bF_{BB^\prime} \right].\label{eq:term2}
\end{align}}In the second equality, we use the ``swap trick'': for operators $M,N\in\linear{\cH_B}$,
it holds that $\tr[MN] = \tr[(M\otimes N)\bF_{BB^\prime}]$. In detail, we choose
$M \equiv N \equiv \tr_A[P_x^{(\theta)}\widetilde{\rho}_{AB}]$. Then
\begin{align*}
&\quad\; \tr_B[MN] \\
&= \tr_{BB^\prime}\left[ \left( M \otimes N\right)\bF_{BB^\prime}\right]  \\
&=  \tr_{BB^\prime AA^\prime}
\left[  \left(P_x^{(\theta)}\widetilde{\rho}_{AB} \otimes 
        P_x^{(\theta)}\widetilde{\rho}_{A^\prime B^\prime}\right)\bF_{BB^\prime}\right] \\
&=  \tr_{BB^\prime AA^\prime}
\left[  \left(P_x^{(\theta)}\otimes P_x^{(\theta)}\right)
        \left(\widetilde{\rho}_{AB} \otimes \widetilde{\rho}_{A^\prime B^\prime}\right)\bF_{BB^\prime}\right].
\end{align*}
In the forth equality, we use the fact that complete set of MUMs forms a conical 2-design (see
Eq.~\eqref{eq:mum-generalization}). Now we compute the two terms given in Eqs.~\eqref{eq:term1} and~\eqref{eq:term2}.
For the first term, we have
\begin{align*}
&\quad\; \tr_{BB^\prime}\tr_{AA^\prime}
\left[ \left(\widetilde{\rho}_{AB} \otimes \widetilde{\rho}_{A^\prime B^\prime} \right) \bF_{BB^\prime} \right] \\
&= \tr_B \left[ \tr_A\left( \widetilde{\rho}_{AB} \right)\tr_A\left(\widetilde{\rho}_{AB}\right)  \right] \\
&= \tr_B \left[ \rho_B^{1/2}\rho_B^{1/2} \right] = 1.
\end{align*}
For the second term, we have
\begin{align*}
&\quad\; \tr_{BB^\prime}\tr_{AA^\prime}
    \left[\bF_{AA^\prime} \left(\widetilde{\rho}_{AB} \otimes \widetilde{\rho}_{A^\prime B^\prime} \right) 
          \bF_{BB^\prime} \right] \\
&= \sum_{ts} \tr_{BB^\prime}\tr_{AA^\prime}
\left[  (\ket{t}\bra{s}\otimes\ket{s}\bra{t}) 
        \left(\widetilde{\rho}_{AB} \otimes \widetilde{\rho}_{A^\prime B^\prime} \right)
        \bF_{BB^\prime} \right] \\
&= \sum_{ts}\tr_{B}\left[\bra{s}\widetilde{\rho}_{AB}\ket{t}\bra{t}\widetilde{\rho}_{AB}\ket{s}\right] \\
&= \tr\left[\widetilde{\rho}_{AB}^2\right] = 2^{-\Collision\left(A\vert B\right)_\rho}.
\end{align*}
Combining these results, we reach at
\begin{equation}\label{eq:complete-eq-2}
(d+1)2^{-\Collision\left(X\vert B\Theta\right)_\omega}
= f(\kappa) + g(\kappa)2^{-\Collision\left(A\vert B\right)_\rho}.
\end{equation}
Rearranging the elements, we get Eq.~\eqref{eq:complete-eq}.
\end{proof}
Following, we discuss several implications of Thm.~\ref{thm:mum-ur}:
its relation to the guessing games,
its relation to the uncertainty relations expressed in bounds on sum of entropies,
and its application in entanglement detection.
These implications can help us gain further intuition about relation~\eqref{eq:complete-eq}.

\subsection{Guessing games}

Note that the conditional collision entropy admits an operational interpretation in terms of guessing games. Now we
consider a game suited to the above MUMs situation. Bob prepares a state $\rho_{AB}$ and sends the $A$ system to Alice.
She measures $A$ in one measurement randomly chosen from the complete set of MUMs, and then tells Bob which measurement
has been performed (index $\theta$). Bob's task is to guess Alice's outcome (index $x$) using the pretty-good
measurements on $B$. Thm.~\ref{thm:mum-ur} can be understood as saying that
Bob's ability to correctly guess the outcome is quantitatively connected to the
pretty-good recoverable entanglement fidelity that can be achieved by Bob. This is summarized as follows.
\begin{lemma}
Let $\{\cP^{(\theta)}\}_{\theta\in[d+1]}$ be a complete set of MUMs on system $A$.
For arbitrary quantum state $\rho_{AB}$, it holds that
\begin{equation}
\sum_{\theta=1}^{d+1}\Ppg\left(X^{(\theta)}{\Big\vert}B\right)_\omega
= f(\kappa) + g(\kappa)\Fpg\left(A{\vert}B\right)_\rho,
\end{equation}
where $\Ppg(X^{(\theta)}{\vert}B)_\omega$ is the pretty-good guessing probability
of state $\omega_{X^{(\theta)}B}$, and $\Fpg(A{\vert}B)_\rho$
is the pretty-good recoverable entanglement fidelity of state $\rho_{AB}$.
\end{lemma}
\begin{proof}
Using the operational interpretations of $\Collision\left(X\vert B\Theta\right)_\omega$ and
$\Collision\left(A\vert B\right)_\rho$, we obtain from Eq.~\eqref{eq:complete-eq-2} that
\[
(d+1)\Ppg\left(X\vert B\Theta\right)_\omega = f(\kappa) + g(\kappa)\Fpg\left(A\vert B\right)_\rho.
\]
Now all we need to show is the following equality
\[
\Ppg\left(X{\vert}B\Theta\right)_\omega 
= \frac{1}{d+1}\sum_{\theta=1}^{d+1}\Ppg\left(X^{(\theta)}{\Big\vert}B\right)_\omega,
\]
where the LHS. is evaluated on state $\omega_{XB\Theta}$, while the RHS. is
evaluated on the states $\omega_{X^{(\theta)}B}$.
This is trivial since $\omega_{XB\Theta}$ is a uniform mixture of the states $\omega_{X^{(\theta)}B}$.
\end{proof}

\subsection{Uncertainty relations expressed in sum of entropies}
Uncertainty relations are commonly expressed as lower bounds on the sum
of entropies of the probability distributions induced by incompatible measurements. 
Using Eq.~\eqref{eq:complete-eq}, we can derive
a uncertainty relation of such kind in terms of sum of collision entropies.
\begin{lemma}\label{lemma:sum-of-entropies}
Let $\{\cP^{(\theta)}\}_{\theta\in[d+1]}$ be a complete set of MUMs on system $A$.
For arbitrary quantum state $\rho_{AB}$, it holds that
\begin{align}
&\quad\; \frac{1}{d+1}\sum_{\theta=1}^{d+1}\Collision\left(X^{(\theta)}{\Big\vert}B\right)_\omega\notag \\
&\geq \log\left(d + 1 \right) - \log\left(f(\kappa) + g(\kappa)2^{-\Collision\left(A\vert B\right)_\rho} \right),
        \label{eq:sum-of-entropies}
\end{align}
where $\omega_{X^{(\theta)}B}$ is defined in Eq.~\eqref{eq:omega-x-theta-b},
coefficients $f(\kappa)$ and $g(\kappa)$ are defined in Eq.~\eqref{eq:mum-generalization}.
\end{lemma}
\begin{proof}
As the $\log$ function is concave, from Eq.~\eqref{eq:classical-quantum-collision} we get
\begin{align*}
&\quad\; \Collision\left(X\vert B\Theta\right)_\omega \\
&\leq \frac{1}{d+1}\sum_{\theta=1}^{d+1}
 - \log\left(\sum_{x=1}^{d}
\tr_B\left\{\tr_A[P_x^{(\theta)}\widetilde{\rho}_{AB}]^2\right\}\right) \\
&= \frac{1}{d+1}\sum_{\theta=1}^{d+1}\Collision\left(X^{(\theta)}{\vert}B\right)_\omega.
\end{align*}
Together with Eq.~\eqref{eq:complete-eq}, we prove this lemma.
\end{proof}
If system $B$ is trivial, Eq.~\eqref{eq:sum-of-entropies} reduces to
\begin{align*}
&\quad\; \frac{1}{d+1}\sum_{\theta=1}^{d+1}\Collision\left(X^{(\theta)}\right)_\omega \\
&\geq \log\left(d + 1 \right) - \log\left(f(\kappa) + g(\kappa)\tr[\rho_A^2] \right).
\end{align*}
This inequality recovers a special case ($\alpha=2$) of \textbf{Theorem 4}
in~\cite{chen2015uncertainty}.

\subsection{Entanglement detection}

Entanglement is appealing feature of quantum mechanics and has been extensively investigated in the past
decades~\cite{horodecki2009quantum}. Entangled states play important roles in many quantum tasks, such as quantum
teleportation~\cite{bennett1993teleporting} and dense coding~\cite{bennett1992communication}. Deciding whether a given
quantum state is entangled is a central problem in quantum information theory and known to be computationally
intractable in general~\cite{gurvits2004classical}. Experimenters often need easy-to-implement methods to verify that
their source is indeed producing entangled states~\cite{guhne2009entanglement}.

Lemma~\ref{lemma:sum-of-entropies} offers a simple strategy for detecting entanglement since it connects entanglement
to uncertainty, while the latter is experimentally measurable. We show that for separable states, the sum of entropies
induced by complete set of MUMs has a larger lower bound, compared to that of entangled states. This bound serves as an
entanglement witness, as any state violates this bound must be necessarily entangled.
\begin{lemma}\label{lemma:entanglement-detection}
Let $\{\cP^{(\theta)}\}_{\theta\in[d+1]}$ be a complete set of MUMs on system $A$,
and let $\{\cQ^{(\theta)}\}_{\theta\in[d+1]}$ be an arbitrary set of $d+1$ measurements on system $B$.
For arbitrary separable quantum state $\rho_{AB}$, it holds that
\begin{align}
&\quad\; \frac{1}{d+1}\sum_{\theta=1}^{d+1}\Collision\left(X^{(\theta)}{\Big\vert} Y^{(\theta)}\right)_\omega\notag \\
&\geq \log\left(d + 1 \right) - \log\left(f(\kappa) + g(\kappa)\right),\label{eq:entanglement-detection-2}
\end{align}
where $\omega_{X^{(\theta)}Y^{(\theta)}}$ is defined as
\begin{equation*}
\omega_{X^{(\theta)}Y^{(\theta)}} = \sum_{x,y=1}^d
  \tr\left[\left(P_x^{(\theta)}\otimes Q_y^{(\theta)}\right)\rho_{AB}\right]
  \proj{x}\otimes\proj{y}.
\end{equation*}
\end{lemma}
\begin{proof}
We use Lemma~\ref{lemma:sum-of-entropies}. It holds that
\begin{align*}
&\quad\; \frac{1}{d+1}\sum_{\theta=1}^{d+1}\Collision\left(X^{(\theta)}{\Big\vert} Y^{(\theta)}\right)_\omega  \\
&\geq \frac{1}{d+1}\sum_{\theta=1}^{d+1}\Collision\left(X^{(\theta)}{\Big\vert} B\right)_\omega \\    
&\geq \log\left(d + 1 \right) - \log\left(f(\kappa) + g(\kappa)2^{-\Collision\left(A\vert B\right)_\rho} \right) \\
&\geq \log\left(d + 1 \right) - \log\left(f(\kappa) + g(\kappa)\right).
\end{align*}
The first inequality follows from
the conditional collision entropy satisfies the data-processing inequality~\cite{mueller-lennert2013quantum},
the second inequality is proved in Eq.~\eqref{eq:sum-of-entropies},
while the last inequality follows from the fact that
all separable states have non-negative collision entropy~\cite{berta2014entanglement}.
\end{proof}
Lemma~\ref{lemma:entanglement-detection} can be used to detect entanglement.
Given a bipartite state $\rho_{AB}\in\density{\cH_A\otimes\cH_B}$, 
Alice performs complete set of MUMs $\cP^{(\theta)}$ on system $A$,
while for each $\theta$ Bob performs a corresponding measurement $\cQ^{(\theta)}$ on system $B$.
They then evaluate the classical collision entropies
$\Collision(X^{(\theta)}{\vert} Y^{(\theta)})$. State $\rho_{AB}$ is entangled if 
\begin{equation*}
\frac{1}{d+1}\sum_{\theta=1}^{d+1}\Collision\left(X^{(\theta)}{\Big\vert} Y^{(\theta)}\right)
< \log\left(d + 1 \right) - \log\left(f(\kappa) + g(\kappa)\right).
\end{equation*}
We remark that the choice of measurements $\cQ^{(\theta)}$ is arbitrary.
For best detection criterion, one can minimize the LHS. of Eq.~\eqref{eq:entanglement-detection-2}
by optimizing over all possible measurements on system $B$.


\section{Uncertainty relations for SIM}\label{sec:sim}

In this section, we show that a similar equality relation in the presence of memory exists 
for generalized symmetric informationally complete measurements.


A set of $d^2$ positive-semidefinite operators $\{P_x\}_{x\in[d^2]}$ in $\cH$ is called a 
\textit{generalized symmetric informationally complete measurement} (SIM) if~\cite{gour2014construction}
\begin{itemize}
\item It is a POVM: $P_x \geq 0$ and $\sum_{x=1}^{d^2} P_x = \1$; and,
\item It is symmetric: $\forall x\in[d^2]$, $\tr[P_x^2]=\eta$,
      $\forall x \neq y$, $\tr[P_xP_y]= \frac{1-\eta d}{d(d^2-1)}$.
\end{itemize}
$\eta$ is the \textit{efficiency parameter} defining the ``type'' of a general SIM, whose range is $1/d^3 < \eta
\leq 1/d^2$. 
There exists a general method to construct the set of \textit{all} generalized SIMs~\cite{gour2014construction}.
In~\cite{graydon2016quantum}, it is proved that every SIM forms a conical 2-design
\begin{equation}\label{eq:sic-generalization}
\sum_{x=1}^{d^2} P_{x} \otimes P_{x} = l(\eta)\1_{AA^\prime} + r(\eta)\bF_{AA^\prime},
\end{equation}
where the coefficients are given by
\begin{equation}\label{eq:coefficients-lr}
l(\kappa) = \frac{1- d\eta}{d^2-1}, \quad
r(\kappa) = \frac{d^3\eta-1}{d(d^2-1)}.
\end{equation}


Let $\cP=\{P_x\}_{x\in[d^2]}$ be a generalized SIM on system $A$,
and $\rho_{AB}$ be quantum state on $AB$. Measuring $\rho_{AB}$ on $A$ by $\cP$, 
we obtain the following classical-quantum state:
\begin{equation}\label{eq:sim-post-state}
  \omega_{XB} = \sum_{x=1}^{d^2} \proj{x} \otimes \tr_A\left[\left(P_x \otimes \1_B\right) \rho_{AB} \right],
\end{equation}
where $\cH_X$ is a $d^2$-dimensional Hilbert space with $\{\ket{x}\}$ being its standard basis. Classical register $X$
indicates which measurement operator is performed; $\tr_A[P_x\rho_{AB}]$ is the post-measurement state (unnormalized)
left in system $B$, conditioned on the measurement operator performed; and $\tr[P_x\rho_{AB}]$ is probability that the
measurement outcome is $x$. We are now ready to present an equality relation for SIM with collision entropy.
\begin{theorem}\label{thm:sim-ur}
Let $\cP=\{P_x\}_{x\in[d^2]}$ be a SIM on system $A$.
For arbitrary quantum state $\rho_{AB}$, it holds that
\begin{equation}\label{eq:complete-sim}
\Collision\left(X{\vert}B\right)_\omega
= - \log\left[ l(\eta) + r(\eta)2^{-\Collision\left(A{\vert}B\right)_\rho} \right],
\end{equation}
where $\omega_{XB}$ is defined in Eq.~\eqref{eq:sim-post-state},
$l(\eta)$ and $r(\eta)$ are defined in Eq.~\eqref{eq:coefficients-lr}.
\end{theorem}
The proof of Thm.~\ref{thm:sim-ur} is identical to that of Thm.~\ref{thm:mum-ur}, with
Eq.~\eqref{eq:sic-generalization} substituted appropriately. 
Now we discuss some consequences of Thm.~\ref{thm:sim-ur}.
When $\eta=1/d^2$, which is the case of symmetric informationally complete measurements~\cite{renes2004symmetric},
Eq.~\eqref{eq:complete-sim} becomes
\[
\Collision\left(X{\vert}B\right)_\omega
= \log\left[ d(d+1) \right] - \log\left[ 1 + 2^{-\Collision\left(A{\vert}B\right)_\rho}\right],
\]
which is exactly the \textbf{Corollary 2} proved in~\cite{berta2014entanglement}.
Trivializing system $B$, Eq.~\eqref{eq:complete-sim} reduces to
\begin{equation}\label{eq:complete-sim-2}
\Collision\left(X\right)_\omega
= \log\frac{d(d^2-1)}{(d^3\eta - 1)\tr[\rho_A^2] + (1-d\eta)d}.
\end{equation}
This is an equality relation for SIM without quantum memory.
Eq.~\eqref{eq:complete-sim-2} recovers and tightens a special case ($\alpha=2$) of \textbf{Proposition 3}
in~\cite{rastegin2014notes}.
We remark that Eq.~\eqref{eq:complete-sim} can also be used to detect entanglement,
using the fact that all separable states have non-negative collision entropy~\cite{berta2014entanglement}.

\section{Conclusions}\label{sec:conclusion}

In summary, we derive several uncertainty relations in the presence of quantum memory for different set of
measurements. Our results are generalizations and extensions of~\cite{berta2014entanglement}. In that paper,
uncertainty relations in the presence of quantum memory was formulated for MUBs using the conditional
collision entropy. In this paper, we prove an equality between the amount of uncertainty for a set of measurements and
the amount of entanglement of the measured state, both of which are quantified by the conditional collision entropy
(Thm.~\ref{thm:mum-ur}). Our result relies on the fact that complete set of mutually unbiased measurements forms a
conical 2-design. Several implications of this equality relation are discussed, among which the entanglement detection
method may be of interest from the experiment's point of view. Using similar techniques, we further prove an
equality relation for generalized symmetric informationally complete measurements
(Thm.~\ref{thm:sim-ur}). By investigating the relation between the construction of complete set of MUMs and the conical
2-design, we derive an interesting equality for arbitrary orthogonal basis of the space of Hermitian, traceless
operators (Lemma~\ref{lemma:operator-basis}). This equality may be helpful for studying conical designs. We hope our
results can shed lights on the study of MUMs and inspire new relations quantifying the relation between
uncertainty and entanglement.

\textit{Acknowledgments.}
This work is supported by the National Natural Science Foundation of China (Grant No. 61300050) and the Chinese
National Natural Science Foundation of Innovation Team (Grant No. 61321491).
\appendix

\section{Correctness of Eq.~\ref{eq:classical-quantum-collision}}\label{appx:classical-quantum-collision}

Here we prove Eq.~\ref{eq:classical-quantum-collision} in the main text.
We shall first compute $\omega_{B\Theta}$:
\begin{align*}
&\quad\; \omega_{B\Theta}= \tr_X \omega_{XB\Theta}    \\
&= \frac{1}{d+1}\sum_{\theta=1}^{d+1} 
    \left(\sum_{x=1}^{d} \tr_A[P_x^{(\theta)}\rho_{AB}]\right) \otimes \proj{\theta}_\Theta \\
&= \rho_B \otimes \frac{1}{d+1}\sum_{\theta=1}^{d+1}  \proj{\theta}_\Theta
= \rho_B\otimes\pi_\Theta.
\end{align*}
Then
\begin{widetext}
\begin{align*}
\Collision\left(X\vert B\Theta\right)_\omega 
&= -\log\tr\left\{ \left(\omega_{B\Theta}^{-1/4}\omega_{XB\Theta}\omega_{B\Theta}^{-1/4}\right)^2 \right\}    \\
&= -\log\tr\left\{\left( 
    \frac{1}{\sqrt{d+1}}\sum_{\theta,x}\proj{x}_X \otimes 
    \rho_B^{-1/4} \tr_A[P_x^{(\theta)}\rho_{AB}] \rho_B^{-1/4}
    \otimes \proj{\theta}_\Theta \right)^2\right\}   \\  
&= -\log\tr\left\{\frac{1}{d+1}\sum_{\theta,x}
    \left(\rho_B^{-1/4}\tr_A[P_x^{(\theta)}\rho_{AB}]\rho_B^{-1/4}\right)^2 
    \right\}  
= -\log\left\{ 
    \frac{1}{d+1}\sum_{\theta,x} \tr_B\left[\tr_A[\widetilde{\rho}_{AB}P_x^{(\theta)}]^2 \right]\right\}.
\end{align*}
\end{widetext}
The third and fourth equality follows from that $\{\ket{x}\}$ and $\{\ket{\theta}\}$ are orthogonal,
the last equality follows because $P_x^{(\theta)}$ only affects on system $A$, while $\rho^{-1/4}_B$
only affects system $B$.

\section{An equality for operator basis}\label{appx:operator-basis}

Based on the fact that a complete set of MUMs forms a canonical 2-design,
we prove an interesting equality for arbitrary orthogonal
basis for traceless hermitian operators acting on $\cH_A$.
This equality has a similar form to the canonical 2-design.
\begin{lemma}\label{lemma:operator-basis}
Let $\{F_k\}_{k\in[d^2-1]}$ be arbitrary orthogonal basis for the
space of Hermitian, traceless operators acting on $\cH_A$. It holds that
\begin{equation}
\sum_{k=1}^{d^2-1} F_{k} \otimes F_{k} = \bF_{AA^\prime} - \frac{1}{d}\1_{AA^\prime},
\end{equation}
where $A^\prime \cong A$ and $\bF_{AA^\prime}$ is the swap operator defined in Eq.~\eqref{eq:swap-operator}.
\end{lemma}
\begin{proof}
The proof relies heavily on the the construction of complete set of MUMs.
Using the relation between $P_x^{(\theta)}$ and $F_x^{(\theta)}$, we have
\begin{eqnarray*}
&~& \sum_{\theta=1}^{d+1}\sum_{x=1}^{d} P_x^{(\theta)} \otimes P_x^{(\theta)} \\
&=& \sum_{\theta=1}^{d+1}\sum_{x=1}^{d}  \left(\frac{1}{d}\1 + t F_x^{(\theta)}\right) 
                              \otimes \left(\frac{1}{d}\1 + t F_x^{(\theta)}\right)    \\
&=& \frac{1+d}{d}\1_{AA^\prime} 
+ \frac{t}{d}\left(\1_A\otimes\widehat{F}_{A^\prime}+\widehat{F}_{A}\otimes\1_{A^\prime}\right) \\
&+& t^2 \sum_{\theta=1}^{d+1}\sum_{x=1}^{d}  F_x^{(\theta)} \otimes F_x^{(\theta)},
\end{eqnarray*}
where $\widehat{F}$ is defined as $\widehat{F} = \sum_{\theta=1}^{d+1}\sum_{x=1}^{d} F_x^{(\theta)}$.
By definition one has $\widehat{F} = 0$. 
Using the relation between $F_x^{(\theta)}$ and $F_{x,\theta}$, 
it can be shown that (through tedious calculation)
\[
\sum_{\theta=1}^{d+1}\sum_{x=1}^{d} F_x^{(\theta)} \otimes F_x^{(\theta)}
= (d+\sqrt{d})^2\sum_{\theta=1}^{d+1}\sum_{x=1}^{d-1} F_{x,\theta} \otimes F_{x,\theta}.
\]
As $F_{x,\theta}$ are just rearrangements of $F_k$, we have
\begin{eqnarray}
&~& \sum_{\theta=1}^{d+1}\sum_{x=1}^{d} P_x^{(\theta)} \otimes P_x^{(\theta)}\notag \\
&=& \left(1+\frac{1}{d}\right)\1_{AA^\prime} 
  + t^2(d+\sqrt{d})^2\sum_{k=1}^{d^2-1} F_{k} \otimes F_{k} \notag   \\
&=& \left(1+\frac{1}{d}\right)\1_{AA^\prime} 
  + \frac{\kappa d - 1}{d-1}\sum_{k=1}^{d^2-1} F_{k} \otimes F_{k},\label{eq:exact-form}
\end{eqnarray}
where the second equality follows from Eq.~\eqref{eq:t-kappa}. Comparing Eq.~\eqref{eq:mum-generalization} and
Eq.~\eqref{eq:exact-form}, we obtain the following equality for arbitrary orthogonal basis for the space of
Hermitian, traceless operators acting on $\cH_A$:
\begin{equation*}
\sum_{k=1}^{d^2-1} F_{k} \otimes F_{k} = \bF_{AA^\prime} - \frac{1}{d}\1_{AA^\prime}.
\end{equation*}
\end{proof}


\begin{thebibliography}{10}

\bibitem{berta2014entanglement}
Mario Berta, Patrick~J Coles, and Stephanie Wehner.
\newblock Entanglement-assisted guessing of complementary measurement outcomes.
\newblock {\em Physical Review A}, 90(6):062127, 2014.

\bibitem{heisenberg1927w}
W~Heisenberg.
\newblock W. heisenberg, z. phys. 43, 172 (1927).
\newblock {\em Z. Phys.}, 43:172, 1927.

\bibitem{robertson1929uncertainty}
Howard~Percy Robertson.
\newblock The uncertainty principle.
\newblock {\em Physical Review}, 34(1):163, 1929.

\bibitem{bialynicki1975uncertainty}
Iwo Bia{\l}ynicki-Birula and Jerzy Mycielski.
\newblock Uncertainty relations for information entropy in wave mechanics.
\newblock {\em Communications in Mathematical Physics}, 44(2):129--132, 1975.

\bibitem{deutsch1983uncertainty}
David Deutsch.
\newblock Uncertainty in quantum measurements.
\newblock {\em Physical Review Letters}, 50(9):631, 1983.

\bibitem{maassen1988generalized}
Hans Maassen and Jos~BM Uffink.
\newblock Generalized entropic uncertainty relations.
\newblock {\em Physical Review Letters}, 60(12):1103, 1988.

\bibitem{coles2017entropic}
Patrick~J Coles, Mario Berta, Marco Tomamichel, and Stephanie Wehner.
\newblock Entropic uncertainty relations and their applications.
\newblock {\em Reviews of Modern Physics}, 89(1):015002, 2017.

\bibitem{spengler2013graph}
Christoph Spengler and Barbara Kraus.
\newblock Graph-state formalism for mutually unbiased bases.
\newblock {\em Physical Review A}, 88(5):052323, 2013.

\bibitem{cerf2002security}
Nicolas~J Cerf, Mohamed Bourennane, Anders Karlsson, and Nicolas Gisin.
\newblock Security of quantum key distribution using d-level systems.
\newblock {\em Physical Review Letters}, 88(12):127902, 2002.

\bibitem{spengler2012entanglement}
Christoph Spengler, Marcus Huber, Stephen Brierley, Theodor Adaktylos, and
  Beatrix~C Hiesmayr.
\newblock Entanglement detection via mutually unbiased bases.
\newblock {\em Physical Review A}, 86(2):022311, 2012.

\bibitem{durt2010mutually}
Thomas Durt, Berthold-Georg Englert, Ingemar Bengtsson, and Karol
  {\.Z}yczkowski.
\newblock On mutually unbiased bases.
\newblock {\em International Journal of Quantum Information}, 8(04):535--640,
  2010.

\bibitem{kalev2014mutually}
Amir Kalev and Gilad Gour.
\newblock Mutually unbiased measurements in finite dimensions.
\newblock {\em New Journal of Physics}, 16(5):053038, 2014.

\bibitem{klappenecker2005mutually}
Andreas Klappenecker and M~Rotteler.
\newblock Mutually unbiased bases are complex projective 2-designs.
\newblock In {\em International Proceedings of the Symposium on Information
  Theory, ISIT}, pages 1740--1744. IEEE, Psicataway, NJ, 2005.

\bibitem{graydon2016quantum}
Matthew~A Graydon and DM~Appleby.
\newblock Quantum conical designs.
\newblock {\em Journal of Physics A: Mathematical and Theoretical},
  49(8):085301, 2016.

\bibitem{tomamichel2015quantum}
Marco Tomamichel.
\newblock {\em Quantum Information Processing with Finite Resources:
  Mathematical Foundations}, volume~5.
\newblock Springer, 2015.

\bibitem{buhrman2008possibility}
Harry Buhrman, Matthias Christandl, Patrick Hayden, Hoi-Kwong Lo, and Stephanie
  Wehner.
\newblock Possibility, impossibility, and cheat sensitivity of quantum-bit
  string commitment.
\newblock {\em Physical Review A}, 78(2):022316, 2008.

\bibitem{nielsen2011quantum}
Michael~A. Nielsen and Isaac~L. Chuang.
\newblock {\em Quantum {Computation} and {Quantum} {Information}}.
\newblock Cambridge University Press, 2011.

\bibitem{chen2015uncertainty}
Bin Chen and Shao-Ming Fei.
\newblock Uncertainty relations based on mutually unbiased measurements.
\newblock {\em Quantum Information Processing}, 14(6):2227--2238, 2015.

\bibitem{horodecki2009quantum}
Ryszard Horodecki, Pawe{\l} Horodecki, Micha{\l} Horodecki, and Karol
  Horodecki.
\newblock Quantum entanglement.
\newblock {\em Reviews of modern physics}, 81(2):865, June 2009.

\bibitem{bennett1993teleporting}
Charles~H Bennett, Gilles Brassard, Claude Cr{\'e}peau, Richard Jozsa, Asher
  Peres, and William~K Wootters.
\newblock Teleporting an unknown quantum state via dual classical and
  einstein-podolsky-rosen channels.
\newblock {\em Physical Review Letters}, 70(13):1895, 1993.

\bibitem{bennett1992communication}
Charles~H Bennett and Stephen~J Wiesner.
\newblock Communication via one-and two-particle operators on
  einstein-podolsky-rosen states.
\newblock {\em Physical Review Letters}, 69(20):2881, 1992.

\bibitem{gurvits2004classical}
Leonid Gurvits.
\newblock Classical complexity and quantum entanglement.
\newblock {\em Journal of Computer and System Sciences}, 69(3):448--484, 2004.

\bibitem{guhne2009entanglement}
Otfried G{\"u}hne and G{\'e}za T{\'o}th.
\newblock Entanglement detection.
\newblock {\em Physics Reports}, 474(1):1--75, 2009.

\bibitem{mueller-lennert2013quantum}
Martin M{\"u}ller-Lennert, Fr{\'e}d{\'e}ric Dupuis, Oleg Szehr, Serge Fehr, and
  Marco Tomamichel.
\newblock On quantum r{\'e}nyi entropies: A new generalization and some
  properties.
\newblock {\em Journal of Mathematical Physics}, 54(12):122203, 2013.

\bibitem{gour2014construction}
Gilad Gour and Amir Kalev.
\newblock Construction of all general symmetric informationally complete
  measurements.
\newblock {\em Journal of Physics A: Mathematical and Theoretical},
  47(33):335302, 2014.

\bibitem{renes2004symmetric}
Joseph~M Renes, Robin Blume-Kohout, Andrew~J Scott, and Carlton~M Caves.
\newblock Symmetric informationally complete quantum measurements.
\newblock {\em Journal of Mathematical Physics}, 45(6):2171--2180, 2004.

\bibitem{rastegin2014notes}
Alexey~E Rastegin.
\newblock Notes on general sic-povms.
\newblock {\em Physica Scripta}, 89(8):085101, 2014.

\end{thebibliography}
\end{document}